\documentclass{article}
\usepackage{spconf,amsmath,graphicx,hyperref}
\newcommand{\V}[1]{\operatorname{vec}\left(#1\right)}

\usepackage{xspace}

\newcommand{\norm}[1]{\left\lVert#1\right\rVert}

\usepackage{tikz}
\usepackage{algorithm}
\usepackage[noend]{algpseudocode}
\usepackage[normalem]{ulem}

\usepackage{amsmath,amsfonts,bm}
\usepackage{amssymb}
\usepackage{centernot}

\usepackage{amsthm}
\theoremstyle{plain}

\newtheorem{corollary}{Corollary}

\newtheorem{theorem}{Theorem}
\newtheorem{definition}{Definition}

\def\1{\bm{1}}

\def\vone{{\bm{1}}}

\def\vh{{\bm{h}}}

\def\vp{{\bm{p}}}

\def\vu{{\bm{u}}}

\def\vw{{\bm{w}}}

\def\vy{{\bm{y}}}

\def\mB{{\bm{B}}}

\def\mH{{\bm{H}}}
\def\mI{{\bm{I}}}

\def\mK{{\bm{K}}}

\def\mM{{\bm{M}}}

\def\mQ{{\bm{Q}}}
\def\mR{{\bm{R}}}
\def\mS{{\bm{S}}}

\def\mU{{\bm{U}}}
\def\mV{{\bm{V}}}
\def\mW{{\bm{W}}}
\def\mX{{\bm{X}}}
\def\mY{{\bm{Y}}}

\DeclareMathAlphabet{\mathsfit}{\encodingdefault}{\sfdefault}{m}{sl}
\SetMathAlphabet{\mathsfit}{bold}{\encodingdefault}{\sfdefault}{bx}{n}

\def\gD{{\mathcal{D}}}

\def\gH{{\mathcal{H}}}

\def\gL{{\mathcal{L}}}
\def\gM{{\mathcal{M}}}
\def\gN{{\mathcal{N}}}

\def\gS{{\mathcal{S}}}

\newcommand{\E}{\mathbb{E}}

\newcommand{\R}{\mathbb{R}}

\title{Differential Privacy of Network Parameters from a System Identification Perspective}
\fiveauthors{\authorcell{Andrew Campbell, Anna Scaglione\sthanks{This work was supported in part by the DoD-ARO under Grant W911NF2210228; and in part by the Director, Cybersecurity, Energy Security, and Emergency Response (CESER) Office of U.S. Department of Energy via the Privacy-Preserving, Collective Cyberattack Defense of DERs Project DEAC02-05CH11231.}\thanks{2025 IEEE. Personal use of this material is permitted. Permission from IEEE must be obtained for all other uses, in any current or future media, including reprinting/republishing this material for advertising or promotional purposes, creating new collective works, for resale or redistribution to servers or lists, or reuse of any copyrighted component of this work in other works.}}
	{Cornell University\\
	ECE Department \\
	New York, USA}}{\authorcell{Hang Liu
    }{University of Macau\\
    SKL-IOTSC \& ECE Department\\
    Macao S.A.R.}}{\authorcell{V\'ictor Elvira
    }{University of Edinburgh\\School of Mathematics\\Edinburgh, Scotland}}{\authorcell{Sean Peisert}{Lawrence Berkeley National Lab\\ Scientific Data Division\\
    Berkeley, USA}}{\authorcell{Daniel Arnold}{Lawrence Livermore National Lab\\ Computational Engineering Division \\Livermore, USA}}
    
\begin{document}
%\ninept
%
\maketitle
\begin{abstract}
This paper addresses the problem of protecting network information from privacy system identification (SI) attacks when sharing cyber-physical system simulations. We model analyst observations of networked states as time-series outputs of a graph filter driven by differentially private (DP) nodal excitations, with the analyst aiming to infer the underlying graph shift operator (GSO). Unlike traditional SI, which estimates system parameters, we study the inverse problem: what assumptions prevent adversaries from identifying the GSO while preserving utility for legitimate analysis. We show that applying DP mechanisms to inputs provides formal privacy guarantees for the GSO, linking the $(\epsilon,\delta)$-DP bound to the spectral properties of the graph filter and noise covariance. More precisely, for DP Gaussian signals, the spectral characteristics of both the filter and noise covariance determine the privacy bound, with smooth filters and low-condition-number covariance yielding greater privacy.

\end{abstract}
\begin{keywords}
Graph Filters, Graph Signal Processing (GSP), Differential Privacy (DP) 
\end{keywords}
\section{Introduction}
\label{sec:intro}
One classic problem in signal processing is system identification (SI). The goal, given observations of a system driven by an excitation signal, is to estimate the system information \cite{zadeh1956identification}. Formally, we denote our observations as $\vy_t\in\R^n$, our excitations by $\vu_t\in\R^n$ and our system by $\gH(\cdot;\bm{\theta})$ where $\gH(\cdot;\bm{\theta})$ is a linear shift invariant (LSI) system with parametrization $\bm{\theta}\in\R^d$. Using this notation we have the following input-output model 
\begin{align}
    \vy_t=\gH(\vu_t;\bm{\theta}). \label{eq:gen_system}
\end{align}
Typically, one first fixes a model class for $\gH(\cdot;\bm{\theta})$ and then applies known excitations in order to estimate the parameter vector $\bm{\theta}$ that best explains the input–output observations \cite{ljung1987theory}. For example in Ho and Kalman's seminal work \cite{ho1966effective}, they determine the impulse response of the deterministic system to then algebraically reconstruct a minimal state-space model by constructing and factoring a Hankel matrix (block matrix of the inputs and outputs). This approach was later extended to handle stochasticity in both N4SID\cite{van1994n4sid} and MOESP\cite{verhaegen1994identification}. While the literature on SI is quite mature and developed, we are specifically interested in systems comprised of graph filters. 
\subsection{Network Identification from Graph Signals}
A graph filter is a LSI map where the invariance is with respect to the graph shift operator (GSO), which is typically represented by the graph Laplacian or adjacency matrix. For example, the order $K$ polynomial graph filter, which is popular in graph convolutional neural networks as the Chebyshev filter \cite{hammond2011wavelets,defferrard2016convolutional},\footnote{The Chebyshev filter is a special type of polynomial filters.} is given by   
\begin{align}
    \mH(\mS;\vh):=\sum_{k=0}^K\vh_k\mS^k.
\end{align} In fact, thinking of graph filters as polynomials of the GSO is particularly useful since, by applying the Cayley-Hamilton theorem \cite{gantmakher2000theory}, all graph filters can be represented by a sufficiently high, possibly infinte, degree polynomial of the GSO. Using this graph filter formulation, \eqref{eq:gen_system} reduces to, where $\bm{\theta}=[\mS,\vh]$,
\begin{align}
    \vy_t= \mH(\mS;\vh)\vu_t. \label{eq:gf_setup}
\end{align} Two classic examples of \eqref{eq:gf_setup} include the diffusion dynamics in social networks and the DC power flow model.  
More generally, the typical SI task is to find the coefficients $\vh$ given a known $\mS$ and known graph filter class. In this regime, if $\vu_t$ is known, identifying $\bm{\theta}$ reduces to a least-squares problem given the linearity.  
However, the problem is made more interesting by instead assuming $\mS$ is unknown and/or by assuming $\vu_t$ is unknown. This setting is referred to as blind SI and is considerably more challenging. In fact, if no assumptions about $\mH$, $\mS$, or $\vu_t$ are made, the problem is ill-posed \cite{segarra2017network}. 

% In \cite{iwata2020graph,ye2025blind,segarra2016blind} $\mS$ is assumed to be known while the input signals $\vu_t$ and coefficients $\vh$ are unknown. The key idea is to leverage the fact that $\mH(\mS)$ and $\mS$ share eigen vectors, meaning that by using the prior of the input signals a tractable optimization problem can be formulated. \cite{iwata2020graph} assumes sparse input signals and utilizes an ADMM solver to find  $\vh$ while \cite{ye2025blind} and \cite{segarra2016blind} use input sparsity but proposes a convex relaxation. On the other hand, \cite{segarra2017network, pasdeloup2017characterization, shafipour2021identifying} assume that both $\mS$ and $\vu_t$ are unknown. In this setting the eigen basis for $\mS$ and $\mH(\mS)$ has to be approximated using the output statics from the $T$ observations and the prior on $\vu_t$. \cite{segarra2017network} assumes stationarity over the output signals and combines this with the input prior to form a sparse recovery problem which was solved with a convex relaxation. \cite{pasdeloup2017characterization} similarity assumes stationarity and output statistics for eigen basis estimation, but specifically studies Diffusion Graph filters utilizing their structure to from a convex set of filters which can then be optimized over. \cite{shafipour2021identifying}, diverges and assumes non-stationary graph signals, utilizing the prior on $\vu_t$ and the outputs to estimate the graph filter and $\mS$.     

We are most interested in the case studied by \cite{segarra2017network, pasdeloup2017characterization, shafipour2021identifying} where both $\mS$ and $\vu_t$ are unknown. Unlike prior work, our goal is to prevent an adversary from solving the identification problem rather than solving it ourselves. Suppose we want to release $\vy_t$ to an analyst so that they can perform some data analysis on the graph signals. We want the analyst to have good results, in the Mean Square Error (MSE) sense, but we do not want them to infer information about $\mS$. More precisely we are interested in the scenario when the input signals have been a priori made Differentially Private (DP). DP is mechanism that adds noise to a release of statistics in a way that can provide probabilistic guarantees that the releases of any two statistics over adjacent datasets are indistinguishable.
% DP is a noise mechanism that perturbs a statistical release in such a way that one can guarantee, in a probabilistic sense, that two statistic releases over adjacent datasets have indistinguishable statistics. 
% The archetypal example is releasing the average income of a group with and without a certain individual. Enough noise is added to this average release to guarantee that the distribution with and without the individual are indiscernible. 
Formally, we say a mechanism is DP if it satisfies the following definition. 
\begin{definition}[$(\epsilon, \delta)-DP$ \cite{dwork2014algorithmic}]
    \label{def:classic_dp}
    Let $\gM$ be a randomized algorithm for publishing a query defined over subsets of dataset $\bm{\gD}$. Then $\forall \bm \mX,\bm \mX'\subset\bm{\gD}$ such that $\mX$ and $\mX'$ are adjacent, $\gM$ is $(\epsilon,\delta)-$DP if $\forall\gS\subset\operatorname{Range}(\gM)$
    \begin{align}
        \Pr\left[\gM(\mX)\in\gS\right]\leq e^\epsilon\Pr\left[\gM(\mX')\in\gS\right]+\delta. \label{eq:dp}
    \end{align}
\end{definition}
Using the notation in Definition \ref{def:classic_dp}, we assume each $\vu_t=\gM(\mX)$ is the output of the DP mechanism. A convenient property of DP is the conservation of privacy under post processing \cite{dwork2014algorithmic}. That is, any post-processing function on $\vy_t$ does not affect the privacy guarantee on the input. However, we focus on the privacy of $\mS$ rather than $\vu_t$. The central question is: without altering $\mH(\mS)$, can the DP signals $\vu_t$ protect $\mS$ and limit an adversary’s ability to identify it? 
We address this by deriving explicit conditions on the noise structure of $\vu_t$ under a Gaussian prior. In fact, to the best of our knowledge, we are the first to provide results on protecting the privacy of a GSO in a graph filter given DP inputs. Before stating the formal problem, we introduce a motivating example.

\subsubsection{Application domain}
% Graph Signal Processing (GSP) has been broadly applied as an approximation of the system that describes many real data that are justified by network dynamics, as they often exhibit the characteristics of low-pass graph filters \cite{LPGP_mag}. Specifically in the case of power systems, we are interested in protecting the system and the consumption data. That is, the consumption data is represented by $\vu_t$ while $\vy_t$ represents the voltage angles (i.e. the states). These states are generated via a graph filter with the consumption data as input. Given that customer consumption data is generally considered to be private, it requires some mechanism to protect customer privacy. While many mechanism exist, DP\footnote{We use DP to refer to both Differential Privacy and Differentially Private.} is the only mechanism with provable guarantees against an arbitrarily powerful adversary and has already seen success consumption data.  \cite{ravi2022differentially,liu2025differentiallyprivatedistributionrelease}. The question thus becomes, given $\vu_t$ has been made DP, is the privacy of the grid structure also protected when we release the states? 
%
Graph Signal Processing (GSP) has been broadly applied as an approximation of the systems describing many real data justified by network dynamics, as they often exhibit low-pass graph-filter characteristics \cite{LPGP_mag}. In power systems, we are interested in protecting the system and the consumption data. The consumption data is $\vu_t$, while $\vy_t$ represents the voltage angles (states). These states are generated via a graph filter with the consumption data as input. Since customer consumption data is generally considered private, a mechanism is required to protect it. While many mechanisms exist, DP\footnote{We use DP to refer to both Differential Privacy and Differentially Private.} is the only mechanism with provable guarantees against an arbitrarily powerful adversary and has already seen success on consumption data \cite{ravi2022differentially,liu2025differentiallyprivatedistributionrelease}. The question thus becomes: given $\vu_t$ has been made DP, is the privacy of the grid structure also protected when we release the states?
% Assuming we are able to generate synthetic consumption data that are differentially private the question this paper is trying to address is if it is possible to rely on the intrinsic identifiability issues of the system to protect it from a curious adversary.  

\section{Problem Formulation}
\label{sec:format}
Consider a graph \(\mathcal{G} = (\mathcal{V}, \mathcal{E})\) with $n$ vertices and GSO denoted by $\mS$. In this paper, we consider undirected graphs that contain no self-loops.
Let $\vy_t\in\R^n$ be a graph signal, $\mH(\mS)$ a graph filter, and $\vu_t$ an excitation. Combining these, we have the standard graph filter denoted as $\vy_t=\mH(\mS)\vu_t.$ 
Furthermore, let $\tilde{\vu}_t$ be a DP excitation with distribution given by $f_{\vu}(\cdot;\bm{\phi})$ such that $
    \tilde{\vy_t}=\mH(\mS)\tilde{\vu_t}.$
Extending this to the matrix format, let $\tilde{\mY},\tilde{\mU}\in\R^{n\times T}$ be formed by column stacking $T$ observations and excitations yielding
    $\tilde{\mY}=\mH(\mS)\tilde{\mU},$
where $\tilde{\mU}$'s probability density function (pdf) is given by $f_{\mU}(\cdot;\bm{\phi})$.  The goal is to show that given $\tilde{\mU}$ is DP, we get a DP guarantee, as in Definition \ref{def:classic_dp}, on release $\tilde{\mY}$ with respect to (w.r.t) $\mS$. Before continuing, we introduce a strictly tighter notion of DP known as probabilistic-DP (PDP) where, PDP $\implies$ DP, but DP$\centernot\implies$PDP \cite{mcclure2015relaxations}.  
\begin{definition}[$(\epsilon,\delta)-PDP$ \cite{machanavajjhala2008privacy}]
    \label{def:pdp}
    Let $\mX,\mX'\subset\mathcal{D}$ be adjacent subsets of dataset $\mathcal{D}$. Let $\mathcal{M}$ denote a query mechanism over $\mathcal{D}$ with output given by $\tilde{q}$ and where $\tilde{q}\sim f(\tilde{q}|\mX)$. We say $\mathcal{M}$ is PDP if 
    \begin{align}
        \Pr \left(\left|\ln\dfrac{f(\tilde{q}|\mX)}{f(\tilde{q}|\mX')}\right|>\epsilon\right)\leq \delta.
    \end{align}
\end{definition}
Additionally we need to define a notion of adjacency. We consider two GSO's $\mS$ and $\mS'$ to be adjacent if they differ by a single edge and
\begin{align}
    \norm{\mS-\mS'}_2\leq \Delta_{\mS}. \label{eq:adjacent}
\end{align} In the traditional DP scenario, adjacent datasets are typically defined as having a single row change with bounded norm \cite{dwork2014algorithmic}, and thus restricting a GSO to bounded single edge changes is the appropriate parallel \cite{wang2013differential,chen2021edge}. Finally, to align our scenario within the PDP framework, we let $\vp(\V{\mY}|\mS;\bm{\phi})$ refer to the pdf of $\mY$ given GSO $\mS$ and parameter vector $\bm{\phi}$ for $f_{\mU}(\cdot;\bm{\phi})$. We then define 
\begin{align}
    \gL_{\mS,\mS'}\left(\V{\mY}\right):=\ln\left(\dfrac{\vp(\V{\mY}|\mS,\bm{\phi})}{\vp(\V{\mY}|\mS',\bm{\phi})}\right).
\end{align} 
Therefore, we want to verify that 
\begin{align}
    \Pr\left[\gL_{\mS,\mS'}\left(\V{\mY}\right)>\epsilon\right]\leq \delta. \label{eq:goal}
\end{align}
In the following section, we verify \eqref{eq:goal} and find explicit conditions assuming $f_{\mU}(\cdot;\bm{\phi})$ is Gaussian. 
% To align this with the standard DP query over a database let 
% \begin{align}
%     \gM(\mL):=\mH(\mL)\tilde{\mB}
% \end{align} be our randomized mechanism over $\mL$ \footnote{And thus $\gG$ since the mapping from $\gG$ to $\mL$ is bijective.}. Therefore, the goal is to determine if the privacy of the underlying graph is protected given the release of a differentially private graph signal. Formally, we want to show that $\gM(\mL)$, for $\mL$ and $\mL'$, satisfies definition \ref{def:dp}. However given that PDP $\implies$ DP we can will instead show that, where $\gL_{\mL,\mL'}{\V{\mX}}$ denotes the log-likelihood of hypothesis $\mL$ over hypothesis $\mL'$, 
% \begin{align}
%     \Pr\left(|\gL_{\mL,\mL'}(\V{\mX})|>\epsilon\right)\leq \delta. \label{eq:our_dp}
% \end{align}

% However, to have a formal notion of differential privacy (DP) we need a formal notion of adjacency. In this paper two graphs $\gG$ and $\gG'$ are adjacent if their respective Laplacians satisfy
%     \begin{align}
%         \norm{\mL-\mL'}_2\leq \Delta_{\mL}.
%     \end{align}

% Under this formulation we define an adjacent $\mL$ by  
% Before continuing, we must first define a notion of adjaceny That is, we want to show that 
% \begin{align}
%     \Pr(\mL)
% \end{align}
% In this setting, we consider a differentially private excitation denoted by $\tilde{\vb}$ with parametrized distribution given by $f_{\vb}(\tilde{\vb};\theta)$.

% Let $\mX\in\R^{n\times T}$ denote the column stacking of $T$ graph signals yielding 

\section{Analysis}
\label{sec:analysis}
% In order to generate a bound in the form of \eqref{eq:goal} we need to formally write out $\gL_{\mS,\mS'}(\V{\mY})$. 
For exposition and the sake of analysis we assume, for now, that $\mH(\mS)$ is invertible. 
% At first this appears to be a strict assumption, 
We will later show how to generalize the analytical results. 
Using the change of variable equation, and letting $\mH_{\mS}^{-1}:=\mI_T\otimes\mH^{-1}(\mS)$ we have 
\begin{align}
    &\vp\left(\V{\mY}|\mS;\bm{\phi}\right)=\dfrac{1}{|\det{\mH(\mS)}|^T}f_{\mU}(\mH_{\mS}^{-1}\V{\mY};\bm{\phi}).\nonumber 
    %\label{eq:og_pdf}
\end{align}
This yields an explicit log-likelihood ratio of
\begin{align}
    &\gL_{\mS,\mS'}(\V{\mY}):= \nonumber\\
    &~~~T\ln\left(\left|\dfrac{\det{\mH(\mS')}}{\det{\mH(\mS)}}\right|\right)+\ln\left(\dfrac{f_{\mU}\left(\mH_{\mS}^{-1}\V{\mY};\bm{\phi}\right)}{f_{\mU}\left(\mH_{\mS'}^{-1}\V{\mY};\bm{\phi}\right)}\right).\label{eq:log-lik} 
\end{align}
We therefore have the following theorem.
\begin{theorem}
\label{th:renyi_dp}
    Let, for $\alpha>1$, $D_{\alpha}(P||Q)$ be the order $\alpha$ R\'enyi Divergence for distributions $P$ and $Q$. Then, for the log-likelihood ratio $\gL_{\mS,\mS'}(\cdot)$ defined in \eqref{eq:log-lik} and where 
    \begin{align}
        P:=\dfrac{f_{\mU}\left(\mH_{\mS}^{-1}\V{\mY};\bm{\phi}\right)}{|\det{\mH(L)}|^T}~~Q:=\dfrac{f_{\mU}\left(\mH_{\mS'}^{-1}\V{\mY};\bm{\phi}\right)}{|\det{\mH(L')}|^T}\nonumber
    \end{align}we have
    \begin{align}
        &\Pr\left(\gL_{\mS,\mS'}(\cdot)>\epsilon\right)\leq \inf_{\alpha}\exp\left\{ (\alpha-1)\left(D_{\alpha}\left(\left.P\right|\left|Q\right.\right)- \epsilon\right)\right\}.\label{eq:renyi_final}
    \end{align}
\end{theorem}
\begin{proof}
    Given that $D_{\alpha}(P||Q)=\dfrac{1}{\alpha-1}\log\int P^{\alpha}Q^{1-\alpha}$, we can apply the Chernoff bound to
    \eqref{eq:log-lik} yielding 
    \begin{align}
        \inf_{\alpha}e^{-\alpha\epsilon}\E\left[e^{\alpha\gL_{\mS,\mS'}(\cdot)}\right]=\inf_{\alpha>1}\exp\left\{ (\alpha-1)\left(D_{\alpha}\left(\left.P\right|\left|Q\right.\right)-\epsilon\right)\right\} \nonumber
    \end{align}
\end{proof}
% Theorem \ref{th:renyi_dp} provides us with a framework for satisfying the DP requirement via the R\'enyi-Divergence of the chosen noise distributions. 
Next, we introduce the following corollary when $\tilde{\mU}$ is given by a Gaussian DP mechanism. We do not assume our DP mechanism is i.i.d. In fact, returning to the power systems example, if $\mS$ represents the power grid and $\mU$ represents the active power, then $\mU$ would have both spatial and temporal correlations as of result of weather patterns and thus the DP mechanism would also require these correlations.    
\begin{corollary}
    \label{cor:gaus}
        In the case when $\V{\tilde{\mU}}\sim\gN\left(\V{\mM},\bm{\Sigma}_{Tn}\right)$, where $\bm{\Sigma}_{Tn}$ is assumed to be invertible and denote $\mH_{\mS}:=\mI_T\otimes\mH(\mS)$, then  let 
        \begin{align}
            \bm{\mu}:= \V{\mH(\mS)\mM} \quad \bm{\mu}':=\V{\mH(\mS')\mM}, \label{eq:cor_means}
        \end{align} 
        \begin{align}
            \mK'&:=\mH_{\mS'}\bm{\Sigma}_{Tn}\mH_{\mS'}^\top~~ \mK:=\mH_{\mS}\bm{\Sigma}_{Tn}\mH_{\mS}^\top \label{eq:cor_cov}
        \end{align}
        and 
        \begin{align}
            \mB_{\alpha}=(1-\alpha)(\mK')^{-1}+\alpha(\mK)^{-1},
        \end{align}        
        then \eqref{eq:renyi_final} is bounded by 
        \begin{align}
            &\inf_{\alpha>1}\left\{\dfrac{1}{2(\alpha-1)}\left[\ln\left(\dfrac{\det(\mB_{\alpha}^{-1})}{\left(\det((\mK'))\right)^{1-\alpha}\left(\det(\mK)\right)^{\alpha}}\right)\right]\right.\nonumber\\
            &~~~~~~\left.+\dfrac{\alpha}{2(\alpha-1)}(\bm{\mu}-\bm{\mu}')^\top\left(\mB_{\alpha}\right)(\bm{\mu}-\bm{\mu}')-\alpha\epsilon\right\}\nonumber\\ \label{eq:cor}
        \end{align}
\end{corollary}
 In the following theorem we 
 % utilize the closed from of the R\'enyi divergence for the gaussian in corollary \ref{cor:gaus} and 
 provide explicit conditions on the covariance matrix $\bm{\Sigma}_{Tn}$. 
\begin{theorem}
    \label{th:guass}
    Using the notation in corollary \ref{cor:gaus}, we further assume, by Lipchitz continuity of the graph filter, that for some constant $\kappa$ \begin{align}
        \norm{\mH(\mS)-\mH(\mS')}\leq \kappa \Delta_{\mS},
    \end{align}
    and 
    \begin{align}
        \gamma:= \sigma_{\min}(\mH(\mS)) ~~\forall \mS \quad \norm{\mH(\mS)}\leq \Gamma ~~\forall \mS.
    \end{align}
    Then, we define
    \begin{align}
        \omega:=\gamma^2\lambda_{\min}(\bm{\Sigma}_{Tn}) ~~~~ \Omega:=\Gamma^2\lambda_{\max}(\bm{\Sigma}_{Tn})\\
        C_{\alpha}(\bm{\Sigma}_{Tn}):=\Omega^{-1}\left(\alpha-(\alpha-1)\dfrac{\Omega}{\omega}\right),
    \end{align}
    which allows us to write, where $D_{\alpha}$ upper bounds the R\'enyi divergence of $\mS$ vs $\mS'$, 
    \begin{align}
        &D_{\alpha}:=\nonumber\\ 
        &\dfrac{Tn}{2(\alpha-1)}\left[\ln\dfrac{\Omega^{\alpha-1}}{C_{\alpha}(\bm{\Sigma}_{Tn})\omega^{-\alpha}}\right]+\dfrac{\alpha(2\alpha-1)\left(\kappa\Delta_{\mS}\norm{\mM}\right)^2}{2(\alpha-1)\omega^{2}}.\nonumber %\label{eq:D}
    \end{align}
    Finally we have that if
    \begin{align}
        D_{\alpha}+\dfrac{1}{\alpha-1}\ln\left(\frac{1}{\delta}\right)\leq \epsilon, \label{eq:th_2_final}
    \end{align}
    % \begin{align}
    % \Delta_R:=\norm{\underbrace{((\mK'))^{-1/2}\mK((\mK'))^{-1/2}}_{\mR}-\mI}_2 \leq \dfrac{2\Gamma\kappa\Delta_{\mS}\norm{\bm{\Sigma}_{Tn}}}{\gamma^2\lambda_{\min}(\bm{\Sigma}_{Tn})}\label{eq:R_def}
    % \end{align}
    % and let
    % \begin{align}
    %     C_{\alpha}(\Delta_R)=\alpha\dfrac{(\alpha-1)+((\alpha-1)\Delta_R)^2}{(1-\Delta_{R})^2(1-(\alpha-1)\Delta_R)^2}.
    % \end{align}
    % then we have that if 
    % %we get the eigen values bc of SPD and PSD 
    % \begin{align}
    %      &\underbrace{\dfrac{Tn}{4\alpha}C_{\alpha}\left(\dfrac{2\Gamma\kappa\Delta_{\mS}\norm{\bm{\Sigma}_{Tn}}}{\gamma^2\lambda_{\min}(\bm{\Sigma}_{Tn})}\right)\left(\dfrac{2\Gamma\kappa\Delta_{\mS}\norm{\bm{\Sigma}_{Tn}}}{\gamma^2\lambda_{\min}(\bm{\Sigma}_{Tn})}\right)^2}_{:=A_1}\nonumber\\~~~~~~&+\underbrace{\dfrac{\alpha\left(\norm{\mM}\kappa\Delta_{\mS}\right)^2}{\lambda_{\min}(\bm{\Sigma}_{Tn})\gamma^2}}_{:=A_2}+
    %      \dfrac{1}{\alpha}\ln(1/\delta)\leq \epsilon, 
    %\end{align}
    then the release of $\mY$ is $(\epsilon,\delta)$-DP w.r.t. $\mS$.
\end{theorem}
Due to space constraints we omit the proof and provide a sketch.
\begin{proof}
    The key idea for bounding the R\'enyi divergence is to utilize the spectral bounds assumed on $\mH(\mS)$. That is, we want to bound $\lambda_{\min}(\mK),\lambda_{\max}(\mK)$ in terms of the assumptions and the spectrum of $\bm{\Sigma}_{Tn}$. Then, we construct $\mR:=\mK^{1/2}(\mK')^{-1}\mK^{1/2}$ and  bound $\lambda_{\max}(\mR)$ and $\lambda_{\min}(\mR)$. Next, we can rewrite $\mB_{\alpha}$ in terms of $\mR$ and bound the det and mean terms using the spectral bounds we derived.    
\end{proof}
% Theorem \ref{th:guass} connects the Lipchitz continuity of Graph filter and the spectrum of both the covariance matrix and the graph filter to the resulting $(\epsilon,\delta)$. 
The key insight from Theorem \ref{th:guass} is the dependence of $\frac{\Omega}{\omega}$. This term is determined by the condition number of $\bm{\Sigma}_{Tn}$ and the spectral bounds of $\mH(\mS)$. The upshot is that one can design $\bm{\Sigma}_{Tn}$ to have the ratio $\frac{\Omega}{\omega}$ close to 1 which does not necessarily depend on the scale of the noise.
%as is the case of the i.i.d. gaussian. 
This means that if a graph filter is particularly smooth and the input signals have a covariance matrix with a low condition number, then a small $\epsilon$ can be achieved w.r.t $\mS$ for free! However, there is no free lunch. In typical DP scenarios a low condition number might be too restrictive especially when correlation through time is high. However, we can still get privacy for free when we consider non-invertible graph filters.

\subsection{Notes on Singular Matrices}
% The astute reader will note that we have not addressed the scenario when $\mH(\mS)$ and $\bm{\Sigma}_{Tn}$ are singular. 
In the general case, where we have density $f_{\mU}(\cdot;\theta)$ and where $\mH(\mS)$ is singular, we start by taking the thin-SVD of $\mH_{\mS}:=\mI_T\otimes \mH(\mS)$ yielding $\mQ_r\bm{\Sigma}_r\mV_r^{\top}=\mH_{\mS}$. Then, if we let $\mV_0$ denote the remaining eigen vectors, $\mW:=\operatorname{span}(\mQ_r)$, and $\bar{\vy}:=\V{\tilde{\mY}}$, we have via the coarea for a linear map, %\underset{\left\{\bar{\vu}\in\mW\right\}}{\vone}
\begin{align}
    &\vp(\bar{\vy}|\mS;\theta)\nonumber\\
     &~~~~=\dfrac{\vone_{\left\{\bar{\vy}\in\mW\right\}}}{\det{\bm{\Sigma}_r}}\int_{\R^{Tn-r}}f_{\mU}\left(\mV_r\bm{\Sigma}_{r}^{-1}\mQ_r^{\top}\bar{\vy}+\mV_0\vw;\theta\right)d\vw.\nonumber
\end{align}
% we define $\bm{\xi}:=\mV_r^{\top}\V{\mU}$ with p.d.f. 
% \begin{align}
%     f_{\bm{\xi}}(\bm{\xi};\theta):=\int_{\R^{Tn-r}}f_{\mU}(\mV_r\bm{\xi}+\mV_{0}\bm{\eta}+\E[\V{\mU}];\theta)d\bm{\eta}.
% \end{align} Therefore, \eqref{eq:og_pdf} can be written as, where $\mW:=\operatorname{span}(\mQ_r)$, $\bar{\vu}:=\V{\mU}$, and $\mH_S:=\mI_T\otimes \mH(\mS)$, 
% \begin{align}
%     &\vp(\V{\mU}|\mS;\theta)\nonumber \\
%     &~~~~=\underset{\left\{\bar{\vu}\in\mW\right\}}{\vone}\dfrac{1}{\det{\bm{\Sigma}_r}}f_{\bm{\xi}}\left(\Sigma_{r}^{-1}\mQ_r^{\top}\left(\bar{\vu}-\mH_S\E[\bar{\vu}]\right);\theta\right). \label{eq_generic_p}
% \end{align} 
Using this, one can again bound the log-likelihood with the R\'enyi divergence and the existing proofs carry through with the appropriate projections. The more interesting case arises when we assume Gaussianity and $\bm{\Sigma}_{Tn}$ is singular. Here, under specific conditions, we can collapse the divergence to 0! We introduce the following corollary. \begin{corollary}
    Let $\mV:=\operatorname{range}(\bm{\Sigma}_{Tn})$, $\mH_{\mS}=\mI_T\otimes \mH(\mS)$, and $\mH_{\mS'}=\mI_T\otimes \mH(\mS')$. Then, if $\mV\subset \ker\left(\mH_\mS-\mH_{\mS'}\right)$ and $\V{\mM}\in \ker\left(\mH_\mS-\mH_{\mS'}\right)$, then the divergence is 0.
\end{corollary} However, in the situation where $\V{\mM}$ and $\bm{\Sigma}_{Tn}$ do not live in the null space, we must apply a projection on to \eqref{eq:cor_means}-\eqref{eq:cor_cov}. Let $\mW:=\operatorname{range}(\mH_{\mS}\bm{\Sigma}^{1/2})$ and $\mW':=\operatorname{range}(\mH_{\mS'}\bm{\Sigma}^{1/2})$ where $\mW=\mW'$. 
% A simple way to pick $\mP$ is to let $\mP:=\mV\mV^\top$, where $\mV$ is an orthonormal basis of $\ker(\mH_{\mS})-\mH_{\mS'}$ 
Then if we let $\mQ$ be an orthonormal basis for $\mW$, \eqref{eq:cor_means}-\eqref{eq:cor_cov} become 
\begin{align}
    \bm{\mu}= \mQ^\top\V{\mH(\mS)\mM}~~\mK=\mQ^\top\mH_\mS\bm{\bm{\Sigma}}_{Tn}\mH_{\mS}^\top\mQ,
\end{align}
and the proofs carry through as before.
\begin{figure}
    \centering
    \includegraphics[width=0.8\linewidth]{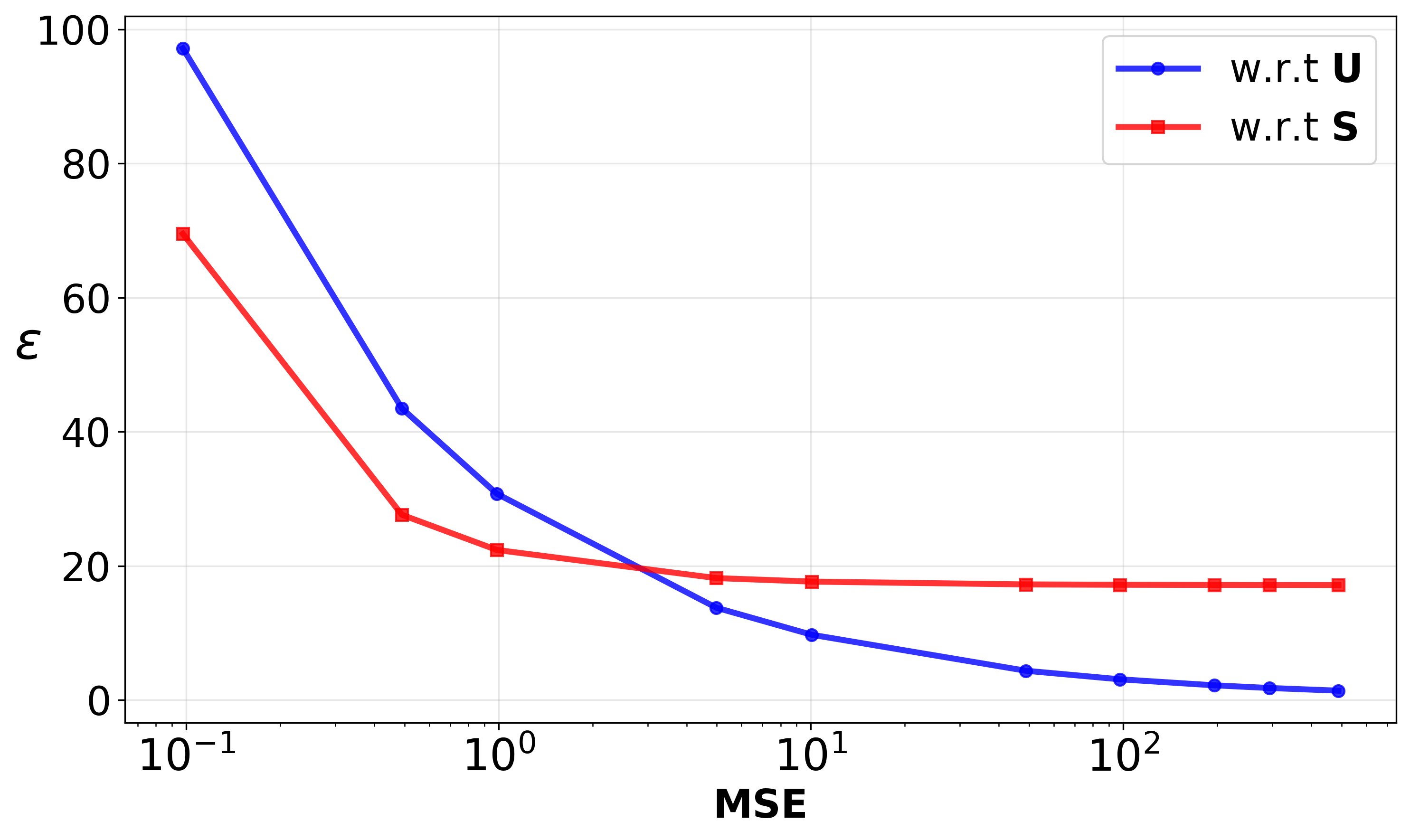}
    \caption{{MSE vs $\epsilon$ w.r.t the GSO and the input.}}
    \label{fig:eps_v_mse}
\end{figure}
\section{Empirical Evaluation}
\label{sec:emperical}
In order to validate the analysis we generated a Erdos-R\'enyi graph and a Gaussian input signal matrix $\tilde{\mU}$ and applied the diffusion-dynamic filter from \cite{LPGP_mag}. Formally,
\begin{align}
    \tilde{\mY}=(\mI-0.01\mS)^{-1}\tilde{\mU},
\end{align} where our graph has 7 vertices and we perform $T=20$ samplings.
We assumed that our input signal was given by some true signal $\mU$ and was made DP using the additive Gaussian mechanism where, $\V{\tilde{\mU}}\sim\gN\left(\V{\mU},\bm{\Sigma}_{Tn}\right)$. In order to generate a time correlated input signal we generated a time covariance matrix $\bm{\Sigma}_T$ via an AR(1) process given by $\bm{\Sigma}_{T}=\sigma^2\rho^{|i-j|}$
where $\rho$ is the correlation and $\sigma$ the noise scale. Then, we formed $\bm{\Sigma}_{Tn}:=\bm{\Sigma}_{T}\otimes \mI_{n}$. Finally we compare the MSE obtained on $\tilde{\mY}$ and compare it to the $\epsilon$ w.r.t $\mS$ and with the $\epsilon$ w.r.t $\mU$. The $\epsilon$ w.r.t $\mS$ is computed using the analytical closed form developed in 
Theorem \ref{th:guass}, while the $\epsilon$ w.r.t $\mU$ is computed using the classical additive gaussian mechanism \cite{dwork2014algorithmic}. The results are presented in Figure  \ref{fig:eps_v_mse}. As can be seen in the figure, the scale of the noise, and thus the MSE for the input signal has a loose dependence. If the input signal has little noise, then there is little privacy for the GSO. However, once there is a small to moderate amount of noise the privacy of the GSO is not directly determined by the scale but rather the structure. In Figure \ref{fig:eps_v_mse} this can be seen where the $\mS$ curve levels out prior to the $\mU$ curve.  
This matches the analytical results which indicate that by carefully designing the covariance structure, the DP of the input signal is enough to protect the GSO.    

\section{Conclusion}
\label{sec:conclusion}
In this paper, we addressed the inverse problem of traditional system identification: protecting the GSO from adversarial inference while maintaining signal utility. We established formal conditions under which DP input signals can guarantee privacy for the GSO. Our main contribution, Theorem \ref{th:guass}, provides explicit bounds relating the privacy parameters $(\epsilon, \delta)$ to the spectral and continuity properties of the graph filter and the spectral properties of the input covariance. The analysis demonstrates that smoother filters with bounded spectrum are better suited for privacy. We can further improve privacy by carefully designing the structure of the covariance matrix to have a low condition number and/or to be indistinguishable with respect to the filter. The key takeaway is that if the graph filter is specially designed and we assume a Gaussian noise mechanism on the input signals, we can effectivity get privacy for ``free" on the graph itself.   

\bibliographystyle{IEEEbib}
\bibliography{strings,refs}

\end{document}